\documentclass[11pt]{article}

\usepackage[margin=1in]{geometry}
\usepackage{amsmath,amssymb,amsthm,mathtools}
\usepackage{booktabs,array,longtable}
\usepackage{microtype}
\usepackage[hidelinks]{hyperref}
\usepackage{xcolor}
\usepackage{enumitem}
\allowdisplaybreaks

\newtheorem{theorem}{Theorem}
\newtheorem{lemma}{Lemma}

\newcommand{\cX}{\mathcal X}
\newcommand{\cY}{\mathcal Y}
\newcommand{\cZ}{\mathcal Z}
\newcommand{\cU}{\mathcal U}
\newcommand{\cV}{\mathcal V}
\newcommand{\supp}{\operatorname{supp}}

\newcommand{\nrect}{\mathrm{nonrect}}
\newcommand{\Rdet}{\mathcal R_{\mathrm{det}}}

\title{Counterexamples to the Markovity Conjecture\\
for the Two-Receiver Broadcast Channel}
\author{Yanxiao Liu\thanks{Yanxiao Liu is with the Department of Electrical and Electronic Engineering, Imperial College London, London, UK. Email: \texttt{y.liu2@imperial.ac.uk}} $\,$
and Mian Huang\thanks{Mian Huang is with the Multimoon Lab, Singapore. Email: \texttt{mhuang5865@gmail.com}}
}

\date{August 12, 2026}

\begin{document}
\maketitle

\begin{abstract}

We present two counterexamples to the \emph{Markovity Conjecture} of Gohari, Liu and Nair (ISIT 2025), which is a structural conjecture concerning the optimizers of the dual functional associated with Marton's inner bound and, if true, would greatly simplify the evaluation of Marton's inner bound.
Both counterexamples are ternary-input broadcast channels with strictly positive transition probabilities and use the same nonrectangular $2\times2$ auxiliary structure.
In each case, we exhibit an explicit non-Markov construction whose objective value is rigorously larger than that achievable by any construction satisfying the conjectured Markov structure.
Both examples are obtained with the assistance of GPT-5.6 Sol and disprove the Markovity Conjecture.
\footnote{The two counterexamples presented in this note were discovered independently by the two authors at approximately the same time (the first by Yanxiao Liu and the second by Mian Huang), before either author became aware of the other's work. The two authors were subsequently put in contact by Prof.~Chandra Nair.}

\end{abstract}

\section{Introduction}

Consider a two-receiver broadcast channel~\cite{cover1972broadcast}, wherein the sender aims to transmit a private input $X\in\cX$ to two receivers $(Y,Z)$ reliably, specified by the marginal channels $W_Y(y| x)$ and $W_Z(z| x)$.  For $\alpha\in[0,1]$ and an input $c=(c_x\in \mathbb{R}^{|\cX|}:x\in\cX)$, define
\begin{equation}
\begin{aligned}
G_\alpha(p;c)
= &-\alpha H(Y)-(1-\alpha)H(Z)
   +I(U;Y)+I(V;Z)-I(U;V)+\mathbb E[c_X],\\
F_\alpha(T;c)
= &\max_{p(u,v,x)}G_\alpha(p;c),
\end{aligned}
\label{eq:functional}
\end{equation}
where the maximization is over distributions satisfying $(U,V)-X-(Y,Z)$.  
This is the sum-rate specialization (i.e., $\lambda=1$ of the weighted sum rate) of the dual characterization considered in \cite{gohari2025conjecture}.

Gohari, Liu and Nair proposed the following structural conjecture for the optimizers of this functional \cite[Conjecture~2]{gohari2025conjecture}: the maximum can always be attained by a distribution satisfying the additional Markov chain
\begin{equation}
U-X-V,
\label{eq:markovity}
\end{equation}
or equivalently, $I(U;V| X)=0$. 
This is referred to as the \emph{Markovity Conjecture} \cite{gohari2025conjecture}.

For the unrestricted maximization in \eqref{eq:functional}, the standard cardinality result gives a useful finite-dimensional reduction.  Indeed,
\[
\begin{aligned}
G_\alpha(p;c)
= &-\alpha I(X;Y)-(1-\alpha)I(X;Z)
   +I(U;Y)+I(V;Z)-I(U;V) -\sum_x d_xp_X(x),
\end{aligned}
\]
where $d_x = \alpha H(Y| X=x) +(1-\alpha)H(Z| X=x) -c_x$. 
Therefore, \cite[Theorem 3]{anantharam2018evaluation} applies with $\delta_0=\delta_1=0$, $\gamma_2=1$, and its parameter $\lambda=\alpha$.  In particular, when $|\cX|=3$, the maximization in \eqref{eq:functional} admits an optimizer satisfying
\begin{equation}
|\cU|+|\cV|\le 4, \quad
X=f(U,V).
\label{eq:card}
\end{equation}
Consequently, the possible cardinality types are a $1\times3$ or $3\times1$ simple branch, or a $2\times2$ deterministic mapping.

Gohari, Liu and Nair further relate the Markovity Conjecture to a \emph{rectangular-mapping} formulation \cite[Proposition~1, Lemma~1 and Remark~2]{gohari2025conjecture}. For a deterministic mapping $X=f(U,V)$, a fiber $f^{-1}(x)$ is called rectangular if it is a Cartesian product of a set of rows and a set of columns.  Their alternate formulation asserts that there is always a maximizer whose nonempty fibers are rectangular.

It is useful, however, to distinguish the literal Markov statement \eqref{eq:markovity} from the deterministic rectangular formulation until the implication between them has been justified for the problem at hand.  In particular, the unrestricted cardinality theorem above guarantees a small deterministic optimizer for the unrestricted optimization, but by itself it does not provide a cardinality or determinism reduction for the distinct constrained optimization over laws satisfying $U-X-V$.  
Section~\ref{sec:bridge} establishes the required bridge directly within the Markov class.  This allows the counterexamples below to refute Conjecture 2 as literally stated, without taking the rectangular-mapping equivalence as an additional assumption.

\section{Bridging Markovity and Rectangularity}
\label{sec:bridge}

In this section, we provide theoretical tools to justify the evaluation of our counterexamples.
For readers who are only interested in the constructions of the counterexamples, please refer directly to Section~\ref{sec::counterexamples}.

For a deterministic encoder $X=f(U,V)$, call $f$ rectangular if every nonempty fiber $f^{-1}(x)$ is a Cartesian product of a set of rows and a set of columns.  Let $\Rdet$ be the class of all finite-alphabet laws with such a deterministic rectangular encoder. 

The numerical certificates for the two counterexamples upper-bound only the deterministic rectangular class $\Rdet$, whereas Conjecture 2 is stated over all stochastic laws satisfying $U-X-V$, with no a priori cardinality or determinism restriction. 
However, there exist two potential issues: 
\begin{enumerate}
    \item arbitrary rectangular auxiliary alphabets must be reduced to the finite map classes covered by the certificate;

    \item the stochastic Markov class must be shown to have finite cardinality and to attain its supremum, and if a Markov law attained the unrestricted optimum, its stochastic encoder would have to be shown deterministic, after which conditional independence and full support force rectangular fibers. 
\end{enumerate}

We provide theoretical tools to resolve the two issues: Lemma~\ref{lem:collapse} addresses the first, and the rest of this section addresses the second.

\begin{lemma}[Three-symbol rectangular collapse]
\label{lem:collapse}
Assume $|\cX|=3$ and $\lambda=1$.  If $X=f(U,V)$ is a finite deterministic rectangular encoder, then there is another rectangular deterministic law with no smaller value of $G_\alpha$ for which either $|\cU|\le2, |\cV|\le2$, or the encoder is one-sided and is covered by the $U=X$ or $V=X$ superposition problem.
\end{lemma}

The proof of Lemma \ref{lem:collapse} is given in Appendix \ref{app::lem_collapse}. 
Lemma \ref{lem:collapse} will help to justify the numerical evaluation in Section \ref{sec::counter_eg_1}; more specifically, Lemma \ref{lem:collapse} justifies why it suffices to optimize only finitely many $2\times2$ rectangular maps when $|\cX|=3$, and it is stronger than merely using \eqref{eq:card}.

\begin{theorem}[Markov class cardinality reduction]
\label{thm:markov-cardinality}
Let $|\cX|=m<\infty$.  For every finite-alphabet law satisfying $U-X-V$, there is another law with the same input marginal, still satisfying $U'-X-V'$, such that $|\cU'|\le m, |\cV'|\le m$ and whose value of $G_\alpha$ is no smaller.  Consequently, the supremum over the literal Markov class is attained.
\end{theorem}

The proof of Theorem~\ref{thm:markov-cardinality} is given in Appendix~\ref{app::markov_cardinality}. By Theorem~\ref{thm:markov-cardinality}, the maximum over the Markov-constrained class is attained with $|\cU|,|\cV|\le |\cX|$.  
Hence, if a Markov law attains $F_\alpha(T;c)$, we may take its auxiliaries to be finite and apply the determinism arguments below.

\begin{lemma}
\label{lem:convex-decomposition}
Fix a finite law $s(u,v)=p(u,v)$ and regard $G_\alpha(s,K)$ as a function of the stochastic encoder $K_{uv}(x)=p(x| u,v)$. 
Then $G_\alpha(s,K)$ is convex in $K$.  If $(s,K)$ is an unrestricted global maximizer, every positive-weight deterministic component in the canonical product decomposition of $K$ is also an unrestricted global maximizer.
\end{lemma}

The proof of Lemma \ref{lem:convex-decomposition} is given in Appendix \ref{app::convex-decomposition}. 
Lemma \ref{lem:convex-decomposition} provides a convex deterministic decomposition, but convexity alone only guarantees that deterministic components of a maximizing stochastic encoder are also globally optimal, and does not suffice for the bridge: those deterministic components need not retain the original Markov condition $U-X-V$. 
The next lemma therefore studies the equality case of this convex decomposition and forces the original maximizing kernel $K$ itself to be deterministic on every active auxiliary cell.

\begin{lemma}[Strict-Jensen determinism]
\label{lem:strict-jensen}
Assume that, for every pair $x\ne x'$, the weighted output-visibility condition
\begin{equation}
\alpha\Vert W_Y(\cdot| x)-W_Y(\cdot| x')\Vert_2^2
+(1-\alpha)\Vert W_Z(\cdot| x)-W_Z(\cdot| x')\Vert_2^2 > 0
\label{eq:visibility}
\end{equation}
holds.  If $(s,K)$ is an unrestricted global maximizer, then $K_{uv}$ is a point mass for every $(u,v)$ in the support of $s$. 
\end{lemma}

The proof of Lemma \ref{lem:strict-jensen} is given in Appendix \ref{app::proof_strict_jensen}. 
It is the key equality-case argument we can utilize: at a global maximum, the convex decomposition from Lemma~\ref{lem:convex-decomposition} must have zero total Jensen gap. 
Strict concavity of output entropy then forces all positive-weight deterministic components to induce the same positively weighted output marginals.  If one active cell randomized between two distinguishable input symbols, two deterministic components differing only at that cell would induce different output marginals, contradicting \eqref{eq:visibility}. 

We now give the key bridge between the literal Markovity condition and the
rectangular formulation used in our numerical certification.

\begin{theorem}
\label{thm:bridge}
Consider a finite-input broadcast channel whose two receiver marginals are
strictly positive and which is irreducible, meaning that no two input symbols
induce the same pair of receiver-marginal rows. 
Fix a potential vector $c$ and assume $0<\alpha<1$.
The following are equivalent:
\begin{enumerate}[label=(\roman*)]
\item an unrestricted global maximizer of $G_\alpha$ satisfies $U-X-V$;
\item some unrestricted global maximizer is realized by a deterministic
rectangular encoder $X=f(U,V)$.
\end{enumerate}
\end{theorem}

The proof of Theorem~\ref{thm:bridge} is given in Appendix~\ref{app::bridge}.
The theorem provides the link between our numerical rectangular-map certificate and the literal Markovity Conjecture. 
Indeed, if a Markov global maximizer existed, Theorem~\ref{thm:markov-cardinality} allows it to be taken with finite auxiliaries.  The strict-Jensen argument then makes its encoder deterministic, while the full-support theorem together with $U-X-V$ forces the resulting deterministic encoder to be rectangular.  Conversely, conditional Markovization of a rectangular global maximizer preserves the deterministic input mapping and cannot decrease the objective, and hence produces a Markov global maximizer.  Therefore,
\[
\begin{gathered}
G_\alpha(p_{\nrect};c)>
\sup_{p\in\Rdet}G_\alpha(p;c)
\quad\Longrightarrow\quad
\text{no rectangular global maximizer}\\
\Longrightarrow\quad
\text{no Markov global maximizer}
\quad\Longrightarrow\quad
\max_{p:\,U-X-V}G_\alpha(p;c)<F_\alpha(T;c),
\end{gathered}
\]
where the last implication uses attainment of the Markov-constrained maximum from Theorem~\ref{thm:markov-cardinality}.

The assumptions in Theorem~\ref{thm:bridge} enter only through these steps: strict positivity is required by the full-support theorem, while irreducibility and $0<\alpha<1$ imply the strict visibility condition used in Lemma~\ref{lem:strict-jensen}.  For $|\cX|=3$, Lemma~\ref{lem:collapse} additionally reduces arbitrary deterministic rectangular architectures to the finite $2\times2$ and one-sided classes used in our numerical certification.

We then provide details of our two counterexamples.
All numerical constants below are interpreted as exact decimal rationals and all logarithms are natural.

\section{Counterexamples}
\label{sec::counterexamples}

\subsection{Counterexample 1}
\label{sec::counter_eg_1}

Let $\cX=\cY=\cZ=\{0,1,2\}$ and define a two-receiver broadcast channel by the conditionally independent coupling
\[
T(y,z| x)=W_Y(y| x)W_Z(z| x),
\]
where
\begin{equation}
W_Y=
\begin{pmatrix}
0.455161376698163&0.373817924805129&0.171020698496708\\
0.758467698206795&0.118614186418326&0.122918115374879\\
0.097309700862227&0.685592439382430&0.217097859755343
\end{pmatrix},
\label{eq:WY}
\end{equation}
\begin{equation}
W_Z=
\begin{pmatrix}
0.424883334335073&0.180039883258550&0.395076782406377\\
0.446958968150175&0.019731351926639&0.533309679923186\\
0.278586009551374&0.595124632894382&0.126289357554244
\end{pmatrix}.
\label{eq:WZ}
\end{equation}
Every channel entry is strictly positive. 
We set
\begin{align}
\alpha &= 0.476379697742613, \\
c & = (0,-0.293271210057665,-0.179371151080177).
\label{eq:alpha-cost}
\end{align}

Take binary $U,V$ with the nonrectangular map and joint law
\begin{equation}
 f(u,v)=
 \begin{pmatrix}0&1\\2&0\end{pmatrix},
\qquad
p_{UV}=
\begin{pmatrix}
0.015284621666388&0.222612061454034\\
0.092282481882513&0.669820834997065
\end{pmatrix}.
\label{eq:candidate}
\end{equation}
The four probabilities in \eqref{eq:candidate} sum exactly to one. 
The induced input law is
\begin{equation*}
p_X=(0.685105456663453,\;0.222612061454034,\;0.092282481882513).
\end{equation*}
The mapping is nonrectangular because the two cells mapped to $x=0$ are diagonal.  Moreover,
\begin{equation*}
I(U;V| X)
=(p_{00}+p_{11})h_2\!\left(\frac{p_{00}}{p_{00}+p_{11}}\right)
=0.073236062052042086\ldots>0,
\end{equation*}
so this explicit auxiliary law does not satisfy $U-X-V$.

A 192-bit outward-rounded MPFR evaluation gives
\begin{equation}
G_\alpha(p_{\nrect};c)
\in
[-1.0280825802447346,\,-1.0280825802447344].
\label{eq:target-interval}
\end{equation}

We then write a $2\times 2$ deterministic map in row-major form as $m=(m_{00},m_{01},m_{10},m_{11})\in\{0,1,2\}^4$.  Exactly $39$ of the $3^4=81$ maps are rectangular.  Relabeling the two rows or the two columns does not change the optimized value, and the $39$ maps form the following $15$ orbits:
\begin{equation}
\begin{gathered}
0000,\ 0011,\ 0012,\ 0022,\ 0101,\ 0102,\ 0121,\ 0122,\\
0202,\ 0211,\ 0212,\ 1111,\ 1122,\ 1212,\ 2222.
\end{gathered}
\label{eq:orbits}
\end{equation}
The two remaining cardinality types in \eqref{eq:card} are the simple branches $U=X,V=\mathrm{const}$ and $V=X,U=\mathrm{const}$.  Hence \eqref{eq:orbits} plus these two branches are exhaustive for the equivalent rectangular formulation of the conjecture.

By Lemma~\ref{lem:collapse}, this finite list in fact covers every finite deterministic rectangular architecture, not only the rectangular maps that happen to lie inside the unrestricted cardinality reduction \eqref{eq:card}.

For a fixed map $m$ and $p=(p_{00},p_{01},p_{10},p_{11})\in\Delta_3$, expansion of the mutual informations yields
\begin{equation}
 g_m(p)
=(1-\alpha)H(Y)+\alpha H(Z)-H(U,Y)-H(V,Z)+H(U,V)+\mathbb E[c_X].
\label{eq:expanded}
\end{equation}
Every probability entering the entropies in \eqref{eq:expanded} is a linear function of $p$.  Thus each map requires a global optimization over only the three-dimensional probability simplex.

The simple branches reduce to global optimization over $q\in\Delta_2$ of
\begin{equation}
\begin{aligned}
g_{U=X}(q)
= &(1-\alpha)H(qW_Y)-(1-\alpha)H(qW_Z)
-\sum_xq_xH(W_Y(\cdot| x))+c^\top q,\\
g_{V=X}(q)
= &\alpha H(qW_Z)-\alpha H(qW_Y)
-\sum_xq_xH(W_Z(\cdot| x))+c^\top q.
\end{aligned}
\label{eq:simple}
\end{equation}

\subsubsection*{Interval branch-and-bound}

For each deterministic map, ordinary numerical optimization only produces feasible values and therefore lower bounds on the class maximum.  It cannot exclude an undiscovered interior maximizer, a boundary maximizer, or a vertex strategy.  To prove a counterexample, we instead need a certified \emph{upper} bound on the objective over the entire probability simplex for every rectangular class.  The interval branch-and-bound calculation supplies exactly this missing direction of inequality.  It is sufficient to certify a separating threshold; the exact optimum of each rectangular class need not be computed. 
We elaborate as follows.

Let $h(t)=-t\log t$, with $h(0)=0$.  On a simplex cell, each linear marginal $q$ has a certified range $[\ell,u]$.  Positive entropy terms are upper-bounded by a tangent at any $t_0>0$:
\begin{equation}
 h(q)\le h(t_0)+h'(t_0)(q-t_0)
=t_0+(-\log t_0-1)q.
\label{eq:tangent}
\end{equation}
Negative entropy terms are upper-bounded using the concavity chord lower bound
\begin{equation}
 h(q)\ge h(\ell)+\frac{h(u)-h(\ell)}{u-\ell}(q-\ell),
\qquad q\in[\ell,u].
\label{eq:chord}
\end{equation}

When $\ell=u$, the corresponding entropy contribution is evaluated directly.  After applying \eqref{eq:tangent}--\eqref{eq:chord}, the objective has an affine upper bound on the cell, so its maximum occurs at a cell vertex.  The initial tetrahedron $\Delta_3$ is recursively partitioned by exact dyadic longest-edge bisection.  Every quantity used in a certified upper bound, every logarithm, and every pruning comparison is evaluated with 192-bit MPFR and directed rounding.  Ordinary binary64 arithmetic is used only to choose a subdivision edge and a convenient tangent location; these choices affect efficiency but not the validity of the global tangent or interval upper bound.  
A cell is discarded only after its certified upper bound is strictly below the exact decimal threshold
\begin{equation}
\tau=-1.0284.
\label{eq:threshold}
\end{equation}
The same construction, on triangular cells, certifies the two functions in \eqref{eq:simple}.

Every one of the $15$ rectangular orbit representatives and both simple branches exhausts its branch-and-bound queue below \eqref{eq:threshold}.  Therefore
\begin{equation}
\sup_{p\in\Rdet}G_\alpha(p;c)<-1.0284.
\label{eq:rect-upper}
\end{equation}
The closest two numerically identified rectangular competitors are the representatives $0212$ and $0211$; both were also independently rerun at 256-bit precision for the same threshold $-1.0284$, with identical certification decisions and subdivision counts.  The conclusion in \eqref{eq:rect-upper} is a threshold certificate, not a claim that the exact rectangular optimum equals $-1.0284$. 

We provide detailed per-class numerical estimates and the high-precision stationary-point values in Appendix~\ref{app:numerical-values}.

For this channel, every marginal transition probability is positive,
\[
0<\alpha<1,
\qquad
\det W_Y
=0.003238095695267287620967209624\ldots\ne0.
\]
Hence the channel is irreducible and satisfies the strict visibility condition.  
Theorem~\ref{thm:bridge}, together with Lemma~\ref{lem:collapse} and \eqref{eq:strict-separation}, now gives the literal conclusion without invoking the rectangular reformulation as an additional premise: for this channel, we have 
\[
\max_{p:\,U-X-V}G_\alpha(p;c)
< F_\alpha(T;c).
\]
In particular, no unrestricted global optimizer satisfies $U-X-V$.

\subsection{Counterexample 2}
We now give a second example, independently certified by exact-rational box bounds and outward-rounded MPFR logarithms.  Set
\begin{equation}
\alpha_2=0.5390910636020981,
\qquad
c^{(2)}=(0,-0.38384735442473183,-0.17505330254261003).
\label{eq:second-alpha-cost}
\end{equation}
Let
\begin{equation}
W_Y^{(2)}=
\begin{pmatrix}
0.36951853194252404&0.3405540519828177&0.28992741607465826\\
0.7644062078066599&0.1243586697987073&0.1112351223946328\\
0.13431916897739554&0.4759630993912907&0.38971773163131376
\end{pmatrix},
\label{eq:second-WY}
\end{equation}
and
\begin{equation}
W_Z^{(2)}=
\begin{pmatrix}
0.4505557317847065&0.21189067443299553&0.33755359378229797\\
0.4099310202669974&0.009896427870306724&0.580172551862695876\\
0.2853422866067611&0.5630135058030781&0.1516442075901608
\end{pmatrix}.
\label{eq:second-WZ}
\end{equation}
The third entry in each row is defined as one minus the first two, so all displayed parameters are exact rationals.  The joint broadcast channel may again be taken as the conditionally independent coupling of its two marginals.

Use the same-style mapping
\begin{equation}
f_2(u,v)=\begin{pmatrix}0&1\\2&0\end{pmatrix}
\label{eq:second-map}
\end{equation}
and the exact masses
\begin{equation}
\begin{aligned}
p_{00}&=0.000587496203024460836846347107816474063425697817587551938108378,\\
p_{01}&=0.0770366662163001055869931249615851671884377052275332967227301,\\
p_{10}&=0.0130832457279635243913988414691657482185168189690949868984089,\\
p_{11}&=1-p_{00}-p_{01}-p_{10}\\
&=0.909292591852711909184761686461432610529619777985784164440752622.
\end{aligned}
\label{eq:second-masses}
\end{equation}
All four masses are strictly positive, while the fiber of $x=0$ consists of the two diagonal cells.  Hence this deterministic law is nonrectangular and cannot satisfy $U-X-V$.
A directed-rounding evaluation gives
\begin{equation}
G_{\alpha_2}(p ;c^{(2)})
\ge -1.07520631015658740221301.
\label{eq:second-lower}
\end{equation}

Lemma~\ref{lem:collapse} reduces arbitrary rectangular auxiliary cardinalities to six genuinely two-sided $2\times2$ orbits and the two one-sided superposition problems.  
The exact-rational interval calculation gives the following rigorous upper bounds.

\begin{table}[h!]
\centering
\small
\renewcommand{\arraystretch}{1.12}
\begin{tabular}{lrrr}
\toprule
Branch & Full-dimensional boxes & Total incl. boundary & Rigorous upper bound \\
\midrule
$0012$ & $215{,}711$ & $259{,}949$ & $-1.07520631031222814389419$ \\
$0102$ & $65{,}488$ & $89{,}445$ & $-1.07520632019360184882798$ \\
$0121$ & $5{,}545$ & $7{,}195$ & $-1.07520652018658400972301$ \\
$0122$ & $3{,}228$ & $4{,}533$ & $-1.07520684613587786915761$ \\
$0211$ & $656{,}940$ & $679{,}714$ & $-1.07520631018369964615548$ \\
$0212$ & $346{,}872$ & $369{,}702$ & $-1.07520631033991559635892$ \\
Superposition $U$ & $55$ & $71$ & $-1.07520796893181037660639$ \\
Superposition $V$ & $30$ & $41$ & $-1.07520799820420479815017$ \\
\bottomrule\\
\end{tabular}
\caption{Directed-interval upper bounds for all deterministic rectangular architectures in the second instance.}
\label{tab:second-certificate}
\end{table}

In particular,
\begin{equation}
\sup_{p\in\Rdet}G_{\alpha_2}(p;c^{(2)})
\le -1.07520631018369964615548.
\label{eq:second-upper}
\end{equation}
Combining \eqref{eq:second-lower} and \eqref{eq:second-upper},
\begin{equation}
G_{\alpha_2}(p ;c^{(2)})
-
\sup_{p\in\Rdet}G_{\alpha_2}(p;c^{(2)})
\ge 2.711224394247\times10^{-11}>0.
\label{eq:second-gap}
\end{equation}

The channel is strictly positive and
\[
\det W_Y^{(2)}
=0.00262234785640323437253459295538\ldots\ne0,
\qquad 0<\alpha_2<1.
\]
Thus Theorem~\ref{thm:bridge} applies. 
We have that also for this channel, 
\[
\max_{p:\,U-X-V}G_{\alpha_2}(p;c^{(2)})
< F_{\alpha_2}(T_2;c^{(2)}),
\]
and the Markovity Conjecture fails.

\section{Concluding Remarks}

We have shown two counterexamples of \cite[Conjecture 2]{gohari2025conjecture}. 
The two use the same irreducible ternary nonrectangular pattern
\[
\begin{pmatrix}A&B\\C&A\end{pmatrix},
\]
but they serve different purposes.  
The first counterexample has a larger guaranteed gap ($\approx 3.174197552654\times10^{-4}$ nats), and its per-class numerical values make the competing architectures transparent. 
In comparison, the second counterexample is specified as an exact-rational dataset and is accompanied by a certificate that treats arbitrary rectangular cardinalities through the row/column-type collapse.  
The strict-Jensen and Markov-cardinality arguments isolate the precise bridge needed to refute the literal stochastic conjecture.

Neither counterexample invalidates Marton's achievable region, and neither resolves the separate Additivity Conjecture \cite[Conjecture 1]{gohari2025conjecture}. 
They show instead that a global optimizer of the Marton dual functional need not admit conditional independence $U-X-V$, even though the objective explicitly penalizes $I(U;V)$.
It remains of interest to determine whether \cite[Conjecture 1]{gohari2025conjecture} holds.

\section*{Acknowledgements}
Yanxiao Liu would like to thank Prof. Chandra Nair for suggesting the use of AI models to search for counterexamples of his conjectures, and Prof. Chandra Nair and Mr. LIU Yi for their help in carefully checking the counterexample.

\bibliographystyle{IEEEtran}
\bibliography{ref.bib}

\appendix

\section{Proof of Lemma \ref{lem:collapse}}
\label{app::lem_collapse}

\begin{proof}
Delete zero-probability rows and columns and write the nonempty fibers as $A_x\times B_x$.  A row of the mapping is a partition of $\cV$ by a subfamily of the at most three sets $B_x$.  If a one-block row occurs, every remaining row has the same complementary partition.  If a three-block row occurs, every row has the same type.  Finally, if two distinct two-block types occur, they may be relabeled as $\{0,1\}$ and $\{0,2\}$.  Then $B_1=B_2=\cV\setminus B_0$; because the rectangles for symbols $1$ and $2$ have overlapping column sets, their row sets are disjoint, excluding the third type $\{1,2\}$.  Hence a two-sided map has at most two row types.  By symmetry, it has at most two column types.

Let $U^*$ be the row type of $U$.  Identical map rows imply that $X$ is a function of $(U^*,V)$, and all terms except $I(U;Y)-I(U;V)$ are unchanged by merging.  Moreover,
\[
\begin{aligned}
G_\alpha(U^*,V,X)-G_\alpha(U,V,X)
&=I(U;V| U^*)-I(U;Y| U^*)\\
&=H(U| U^*,Y)-H(U| U^*,V)\ge0,
\end{aligned}
\]
because, conditionally on $U^*$, one has $U-V-Y$.  Merging identical column types $V\mapsto V^*$ similarly gives, at $\lambda=1$,
\[
G_\alpha(U^*,V^*,X)-G_\alpha(U^*,V,X)
=H(V| V^*,Z)-H(V| V^*,U^*)\ge0.
\]
Thus every two-sided rectangular architecture collapses to a $2\times2$ one.  If only one row or one column type remains, data processing reduces it to the corresponding one-sided superposition problem.
\end{proof}

\section{Proof of Theorem \ref{thm:markov-cardinality}}
\label{app::markov_cardinality}
\begin{proof}
Write the Markov law as $p_X(x)A(u| x)B(v| x)$.  First hold $p_X$ and $B$ fixed.  Put $w_u=p(u)$ and $r_u(x)=p(x| u)$.  Then $\sum_u w_ur_u=p_X$, and the only $U$-dependent part of the objective is
\[
I(U;Y)-I(U;V)
=\sum_u w_u\Bigl[D(r_uW_Y\Vert p_Y)-D(r_uB\Vert p_V)\Bigr].
\]
Restrict the linear moment problem first to the finite posterior points $r_u$ already present in the given law, and optimize only their weights subject to the barycenter constraint.  The original weights are feasible.  An extreme feasible weight vector has at most $m$ positive entries: if more than $m$ positive-weight posteriors were present, their weights could be perturbed in both directions while preserving the barycenter.  Hence $U$ may be reduced to at most $m$ values without decreasing the objective.

Now hold the reduced $A$ and $p_X$ fixed and apply the same argument to
\[
I(V;Z)-I(U;V)
=\sum_v p(v)\Bigl[D(p_{X| v}W_Z\Vert p_Z)-D(p_{X| v}A\Vert p_U)\Bigr].
\]
This reduces $V$ to at most $m$ values without changing $A$ or undoing the first bound.  Compactness and continuity give attainment; boundary input distributions are handled on their support.
\end{proof}

\section{Proof of Lemma \ref{lem:convex-decomposition}}
\label{app::convex-decomposition}
\begin{proof}
For fixed $s$, the output marginals and the channels from $U$ to $Y$ and from $V$ to $Z$ are affine in $K$.  Negative entropy is convex, mutual information is convex in a channel with fixed input law, $-I(U;V)$ is constant, and the potential term is affine.  Thus $G_\alpha(s,K)$ is convex.

For every map $f:\supp s\to\cX$, define
\[
\rho_f=\prod_{(u,v)\in\supp s}K_{uv}(f(u,v)).
\]
Then $K$ is the convex combination of the deterministic kernels $K_f(x| u,v)=\mathbf 1\{x=f(u,v)\}$ with weights $\rho_f$.  If $G_\alpha(s,K)=F_\alpha(T;c)$, convexity and global optimality give
\[
F_\alpha(T;c)
\le\sum_f\rho_fG_\alpha(s,K_f)
\le\sum_f\rho_fF_\alpha(T;c)=F_\alpha(T;c).
\]
Therefore every component with $\rho_f>0$ has the global value.
\end{proof}

\section{Proof of Lemma \ref{lem:strict-jensen}}
\label{app::proof_strict_jensen}

\begin{proof}
Use Lemma~\ref{lem:convex-decomposition}.  Equality holds in the convex decomposition of the objective.  Every Jensen gap contributed by its convex terms is nonnegative, so equality holds separately for $-H(Y)$ whenever $\alpha>0$ and for $-H(Z)$ whenever $1-\alpha>0$.  Strict concavity of entropy therefore forces all positive-weight deterministic components to have the same $Y$-marginal in the first case and the same $Z$-marginal in the second.

If an active cell $(u,v)$ assigned positive probability to distinct inputs $x,x'$, the product decomposition would contain two positive-weight deterministic maps that agree everywhere except at this cell.  Equality of their output marginals would give
\[
s(u,v)\bigl(W_Y(\cdot| x)-W_Y(\cdot| x')\bigr)=0
\]
whenever $\alpha>0$, and the analogous equality for $Z$ whenever $1-\alpha>0$.  Since $s(u,v)>0$, this contradicts \eqref{eq:visibility}.
\end{proof}

\section{Proof of Theorem \ref{thm:bridge}}
\label{app::bridge}

\begin{proof}
Suppose first that a deterministic rectangular global maximizer exists.  Replace its conditional law in every fiber by
\[
p(x)p(u| x)p(v| x).
\]
Because each fiber is rectangular, this conditional Markovization remains inside the same fiber, so $X$ is still a function of $(U,V)$.  It preserves $p_X,p_{UX},p_{VX}$ and hence all objective terms except $-I(U;V)$.  Therefore
\[
G_\alpha(\mathsf Mp;c)-G_\alpha(p;c)=I_p(U;V| X)\ge0.
\]
Global optimality forces equality, and $\mathsf Mp$ is a Markov global maximizer.

Conversely, assume a Markov law attains the unrestricted value.  Theorem~\ref{thm:markov-cardinality} permits finite auxiliaries.  Irreducibility and $0<\alpha<1$ imply \eqref{eq:visibility}, so Lemma~\ref{lem:strict-jensen} forces the encoder itself to be deterministic, say $X=f(U,V)$.  At its fixed input marginal, this law must maximize
\[
I(U;Y)+I(V;Z)-I(U;V),
\]
because every other term of $G_\alpha$ depends only on $p_X$.  The full-support theorem of Gohari--El Gamal--Anantharam \cite[Theorem 1]{gohari2014marton} then gives $p(u,v)>0$ on the entire active grid after unused labels are deleted.  Finally,
\[
p(u,v| x)=p(u| x)p(v| x)
\]
and deterministic encoding imply
\[
f^{-1}(x)=\supp p(U| x)\times\supp p(V| x)
\]
for every nonempty fiber.  Thus $f$ is rectangular.
\end{proof}

\section{Numerical values of the rectangular classes} 
\label{app:numerical-values}
For a representative map $m=(m_{00},m_{01},m_{10},m_{11})$, let
\[
M_m:=\max_{p_{UV}\in\Delta_3}G_\alpha(p_{UV};m,c).
\]
The column $\widehat M_m$ below is the independently cross-checked numerical estimate of $M_m$.  The MPFR branch-and-bound calculation was a threshold certificate rather than an optimizer: it rigorously certified
\[
M_m<-1.0284
\]
for every class, but it did not by itself output the exact value of $M_m$.  All logarithms are natural, so the values are in nats.

\begin{table}[h!]
\centering
\scriptsize
\renewcommand{\arraystretch}{1.14}
\setlength{\tabcolsep}{4.2pt}
\begin{tabular}{c c r l r r}
\toprule
Class & Representative $f(u,v)$ & $\widehat M_m$ & Numerical maximizing regime & Nodes & Depth \\
\midrule
$0000$ & $\begin{psmallmatrix}0&0\\0&0\end{psmallmatrix}$
& $-1.033950244795077$ & constant $X=0$ & $62{,}470$ & $27$ \\
$0011$ & $\begin{psmallmatrix}0&0\\1&1\end{psmallmatrix}$
& $-1.029164107755570$ & $U=X$, support $\{0,1\}$ & $156{,}435$ & $33$ \\
$0012$ & $\begin{psmallmatrix}0&0\\1&2\end{psmallmatrix}$
& $-1.029164107755570$ & boundary: $U=X$, support $\{0,1\}$ & $46{,}199$ & $33$ \\
$0022$ & $\begin{psmallmatrix}0&0\\2&2\end{psmallmatrix}$
& $-1.033950244795077$ & constant $X=0$ & $21{,}793$ & $27$ \\
$0101$ & $\begin{psmallmatrix}0&1\\0&1\end{psmallmatrix}$
& $-1.033950244795077$ & constant $X=0$ & $27{,}988$ & $27$ \\
$0102$ & $\begin{psmallmatrix}0&1\\0&2\end{psmallmatrix}$
& $-1.030243577369243$ & boundary: $V=X$, support $\{0,2\}$ & $29{,}331$ & $30$ \\
$0121$ & $\begin{psmallmatrix}0&1\\2&1\end{psmallmatrix}$
& $-1.030715605716285$ & boundary: $V=X$, support $\{1,2\}$ & $19{,}773$ & $30$ \\
$0122$ & $\begin{psmallmatrix}0&1\\2&2\end{psmallmatrix}$
& $-1.032393431756881$ & boundary: $U=X$, support $\{1,2\}$ & $15{,}181$ & $27$ \\
$0202$ & $\begin{psmallmatrix}0&2\\0&2\end{psmallmatrix}$
& $-1.030243577369243$ & $V=X$, support $\{0,2\}$ & $74{,}643$ & $30$ \\
$0211$ & $\begin{psmallmatrix}0&2\\1&1\end{psmallmatrix}$
& $-1.028812653160046$ & full-support interior point & $145{,}242$ & $34$ \\
$0212$ & $\begin{psmallmatrix}0&2\\1&2\end{psmallmatrix}$
& $-1.028705925572404$ & full-support interior point & $143{,}032$ & $34$ \\
$1111$ & $\begin{psmallmatrix}1&1\\1&1\end{psmallmatrix}$
& $-1.040946155452082$ & constant $X=1$ & $19{,}814$ & $23$ \\
$1122$ & $\begin{psmallmatrix}1&1\\2&2\end{psmallmatrix}$
& $-1.032393431756881$ & $U=X$, support $\{1,2\}$ & $28{,}174$ & $26$ \\
$1212$ & $\begin{psmallmatrix}1&2\\1&2\end{psmallmatrix}$
& $-1.030715605716285$ & $V=X$, support $\{1,2\}$ & $40{,}268$ & $29$ \\
$2222$ & $\begin{psmallmatrix}2&2\\2&2\end{psmallmatrix}$
& $-1.053609853725562$ & constant $X=2$ & $7{,}508$ & $20$ \\
\midrule
$U=X$ & simple branch
& $-1.029164107755570$ & $p_X=(0.6743247741,0.3256752259,0)$ & $76$ & $9$ \\
$V=X$ & simple branch
& $-1.030243577369243$ & $p_X=(0.8119058198,0,0.1880941802)$ & $43$ & $9$ \\
\bottomrule\\
\end{tabular}
\caption{Cross-checked per-class numerical maxima and the node/depth data from the completed 192-bit outward-rounded threshold certificates.  For every row, the certified conclusion is $M_m<-1.0284$.}
\label{tab:rectangular-class-values}
\end{table}

The two genuine interior stationary points were refined at 80-digit precision:
\[
\begin{aligned}
M_{0211}
&\approx -1.0288126531600463791578719687926416013,\\
M_{0212}
&\approx -1.0287059255724039042921401582684364633.
\end{aligned}
\]
The latter remains the largest rectangular value found.  Comparing it with the explicit nonrectangular value
\[
G_{\mathrm N}
=-1.0280825802447345165424074624039378864\ldots
\]
gives the numerical gap
\[
G_{\mathrm N}-\widehat M_{0212}
\approx 6.2334532766938775\times 10^{-4}.
\]

Combining \eqref{eq:target-interval} and \eqref{eq:rect-upper} gives the strict separation
\begin{equation}
\begin{aligned}
G_\alpha(p_{\nrect};c)
&\ge -1.0280825802447346\\
&>-1.0284\\
&>\sup_{p\in\Rdet}G_\alpha(p;c).
\end{aligned}
\label{eq:strict-separation}
\end{equation}
The guaranteed margin between the explicit nonrectangular value and the separating threshold is at least
\begin{equation}
3.174197552654\times10^{-4}\quad\text{nats}.
\end{equation}

\end{document}